\documentclass[12pt]{amsart}
\usepackage{amssymb,latexsym,amsmath,amscd,graphicx,url,color}
\usepackage{enumitem, comment, pifont, csquotes}
\setitemize{noitemsep,topsep=0pt,parsep=0pt,partopsep=0pt}
\newcommand{\cmark}{\ding{51}}
\newcommand{\xmark}{\ding{55}}

\theoremstyle{plain}

\newtheorem{thm}{Theorem}[section]
\newtheorem{prop}[thm]{Proposition}
\newtheorem{cor}[thm]{Corollary}
\newtheorem{lemma}[thm]{Lemma}

\theoremstyle{definition}
\newtheorem{defn}[thm]{Definition}
\newtheorem{rmk}[thm]{Remark}

\newtheorem{ex}[thm]{Example}

\newcommand{\F}{{\mathcal F}}
\newcommand{\FF}{{\mathbb F}}

\newcommand{\M}{{\mathcal M}}

\newcommand{\NN}{{\mathbb N}}

\newcommand{\ZZ}{{\mathbb Z}}

\newcommand{\HF}{\mathrm{HF}}

\DeclareMathOperator{\reg}{reg}

\DeclareMathOperator{\rowsp}{rowsp}
\DeclareMathOperator{\sd}{sd}
\DeclareMathOperator{\Syz}{Syz}
\DeclareMathOperator{\ff}{ff}
\DeclareMathOperator{\ttop}{top}
\DeclareMathOperator{\maxGB}{max.GB.deg}
\DeclareMathOperator{\Triv}{Triv}
\DeclareMathOperator{\chara}{char}

\title{Solving degree, last fall degree, and related invariants}

\author{Alessio Caminata}
\address{Alessio Caminata, Dipartimento di Matematica, Università di Genova, via Dodecaneso 35, 16146 Genova, Italy}  
\email{caminata@dima.unige.it}

\author{Elisa Gorla}
\address{Elisa Gorla, Institut de Math\'ematiques, Universit\'e de Neuch\^atel, 
Rue Emile-Argand 11, 2000 Neuch\^atel, Switzerland}  
\email{elisa.gorla@unine.ch}

\thanks{ 
	\\ The first author is supported by the Italian PRIN2020, Grant number 2020355B8Y \enquote{Squarefree Gr\"obner degenerations, special varieties and related topics}.}

\begin{document}
\maketitle 

\begin{abstract}
In this paper we study and relate several invariants connected to the solving degree of a polynomial system. This provides a rigorous framework for estimating the complexity of solving a system of polynomial equations via Gr\"obner bases methods. Our main results include a connection between the solving degree and the last fall degree and one between the degree of regularity and the Castelnuovo--Mumford regularity.
\end{abstract}

\section*{Introduction}

As computational problems can often be modelled via polynomial equations, several security estimates in cryptography depend on the complexity of polynomial system solving. Within public key cryptography, polynomial system solving plays a crucial role in index-calculus on Abelian varieties \cite{Gaudry} as well as in multivariate cryptography \cite{DPS20}. The relation-collection phase of index calculus algorithms on elliptic and hyperelliptic curves relies on the ability to solve multivariate polynomial systems over a finite field. Being able to accurately predict the complexity of solving these systems has an impact on the balancing of complexities between the relation-collection and the linear-algebra phase, and consequently on the complexity of the whole algorithm. Multivariate cryptography is one of the current proposals for post-quantum cryptography and it offers short digital signatures which are believed to be quantum-resistant. Forging a signature requires finding a solution to a system of multivariate polynomial equations. Therefore, security estimates in multivariate cryptography rely on estimates for the complexity of solving certain polynomial systems. 

The solutions of a system of polynomial equations over a finite field can be computed in polynomial time from a lexicographic Gr\"obner basis of the system. Nowadays, the most efficient algorithms to compute Gr\"obner bases belong to the family of linear-algebra-based algorithms, including F4/F5 and the family of XL Algorithms \cite{Fau99,Fau02,XL,MXL}. The complexity of these algorithms is bounded from above by a known function of the solving degree. Therefore, an estimate on the solving degree of a system of equations directly translates into an upper bound on the complexity of solving the system.

Computer experiments are often used to provide an indication on how the solving degree scales with the parameters of the system. In multivariate cryptography, however, the complexity of solving the polynomial system is huge by design and there is a large gap between the size of the parameters for which we can compute the solving degree experimentally and the proposed system parameters. Because of this large gap, the security estimates obtained from the experimental results by extrapolation have limited reliability. In such a scenario, proven security estimates become essential.

Unfortunately however, finding the solving degree of a system without computing its Gr\"obner basis is often hard. This motivated the introduction of several algebraic invariants related to the solving degree. The main ones are the first \cite{DG10,DS13} and last fall degrees \cite{Huang2018, HKY15}, the degree of regularity \cite{BFS04,Bar04}, and the Castelnuovo--Mumford regularity \cite{CCG, CG21, CG22}. 
These invariants are routinely used within the cryptographic community in order to estimate the solving degree. Relating them with the solving degree and with each other is therefore essential for producing accurate and provable security estimates.

Several connections between these invariants, both heuristic and proven, are known.
In \cite{CG21} we showed that, under suitable assumptions, the Castelnuovo--Mumford regularity of the homogenization of a system is an upper bound for its solving degree. In \cite{Huang2018} the authors outlined an algorithm to compute the solutions of a polynomial system and provided an upper bound on its complexity in terms of its last fall degree. Later in \cite{Gorla2018} it was shown that the last fall degree is a lower bound for the solving degree. Only recently in \cite{ST,T} it was shown that, under suitable assumptions, the solving degree is bounded from above by twice the degree of regularity.
Finally, while the first fall degree is widely used as a proxy for the solving degree, their relation is only heuristic. Several authors in fact assume that the solving degree and the first fall degree are not far apart. This is usually referred to as {\em first fall degree assumption}. The intuition behind this assumption is that, once some degree falls appear, new polynomials are added to the computation and this makes the computation terminate soon after. There are examples for which this assumption does not hold, see e.g. \cite{DS13}, however many authors believe that it is satisfied with high probability, see e.g. \cite{HPST}.

In this paper, we explore the relations between these invariants. We discuss the new equivalent definitions for the last fall degree and provide two new ones, most notably one that involves the new concept of degree falls. We also look at the first fall degree and show by means of examples that it may be arbitrarily larger or smaller than either of the other invariants. 

Our main theorem, Theorem~\ref{sd&lfd}, shows that, for any degree-compatible term order, the solving degree of a system is the maximum between its last fall degree and the largest degree of an element in a reduced Gr\"obner basis of the system. This provides a proof for the intuitive fact that the two key ingredients in determining the solving degree of a system are the degrees of the elements in its reduced Gr\"obner basis and the degree falls. Notice that, while it is clear that these two ingredients play a role in determining the solving degree, until now it was not known that these are the only two relevant factors. 

Another main result, Theorem~\ref{thm:reg}, relates the degree of regularity of a system with the Castelnuovo--Mumford regularity of its homogenization. While for a homogeneous system the degree of regularity, the Castelnuovo--Mumford regularity, and the solving degree of the system often coincide, the relation between these invariants is much more complicated for a non-homogeneous system. If the system is not homogeneous, both the degree of regularity and the Castelnuovo--Mumford regularity are algebraic invariants of associated homogeneous systems. Their relation, however, depends on a number of factors and was until now unclear. In this paper, we show that the degree or regularity is smaller than or equal to the Castelnuovo--Mumford regularity. This is consistent with the known fact that the solving degree is bounded from above by the Castelnuovo--Mumford regularity \cite{CG21} and by twice the degree of regularity \cite{ST,T}.

The paper is structured as follows.
In Section~\ref{sect:prelim} we introduce all the invariants that we discuss in the paper. In Section~\ref{sect:lfd} we compare two equivalent definitions of last fall degree. In Definition~\ref{defn:degfalls} we formalize the concept of degree falls and in Theorem~\ref{d_f} we prove that the last fall degree is the largest degree in which a degree fall occurs. Finally, in Theorem~\ref{gbasis} we show that, for any degree-compatible term order, the last fall degree is equal to the largest degree in which a degree fall associated to an element in a Gr\"obner basis occurs. Section~\ref{section:sdeLFD} contains the main result of this paper. In Theorem~\ref{sd&lfd} we prove that the solving degree of a system is the maximum between its last fall degree and the degrees of the elements in its reduced Gr\"obner basis, under the assumption that the term order is  degree-compatible. Section~\ref{sect:ffd} focuses on the first fall degree. We show by means of examples that the difference between the first fall degree and any other of the invariants that we have mentioned can be either positive or negative and arbitrarily large in absolute value. 
Finally, in Section~\ref{sect:reg} we discuss the relation between the degree of regularity and the Castelnuovo--Mumford regularity. In Theorem~\ref{thm:reg} we show that the degree of regularity of a system is smaller than or equal to the Castelnuovo--Mumford regularity of its homogenization. 

{\bf Acknowledgements:} We wish to express our gratitude to the referee for the careful reading and the insightful comments, which allowed us to improve the manuscript.

\section{Preliminaries}\label{sect:prelim}

Throughout the paper, we let $k$ be a field and denote by $R=k[x_1,\dots,x_n]$ the polynomial ring in $n$ variables with coefficients in $k$. We let $\mathcal{F}=\{f_1,\dots,f_r\}\subseteq R$ be a system of polynomial equations and let $(\mathcal{F})=(f_1,\ldots,f_r)\subseteq R$ be the ideal generated by the equations of $\mathcal{F}$. For $f\in R$, let $f^{\ttop}$ denote the homogeneous component of $f$ of largest degree, i.e., if $f=\displaystyle\sum_{a_1,\ldots,a_n\geq 0}\alpha_{a_1,\ldots,a_n}x_1^{a_1}\cdots x_n^{a_n}$ with $\alpha_{a_1,\ldots,a_n}\in k$, then $$f^{\ttop}=\sum_{a_1+\ldots+a_n=\deg(f)}\alpha_{a_1,\ldots,a_n}x_1^{a_1}\cdots x_n^{a_n}.$$
We denote by $\mathcal{F^{\ttop}}$ the system $\{f_1^{\ttop},\dots,f_r^{\ttop}\}\subseteq R$. Let $t$ be a new variable and denote by $f^h$ the homogenization of a polynomial $f$ with respect to $t$. Let $\mathcal{F}^h$ be the system $\{f_1^h,\dots,f_r^h\}\subseteq R[t]$.

For an $i\in\mathbb{N}$, we denote by $R_i$ the $k$-vector space generated by the monomials of degree $i$ of $R$. For an ideal $I\subseteq R$, let $I_i=I\cap R_i$. Denote by $R_{\leq i}=\sum_{j\leq i} R_j$ and by $I_{\leq i}=I\cap R_{\leq i}$. 

In this paper we study several invariants used to estimate the complexity of solving a polynomial system via Gr\"obner bases methods. The main such invariant is the solving degree, which by definition measures the complexity of computing a Gr\"obner basis using a linear algebra algorithm that follows the strategy proposed by Lazard in~\cite{Lazard}. We now briefly recall what such an algorithm does.

Fix a term order $\sigma$ on $R$ and a degree $d\geq 1$. Let
$$\M_{\leq d}=\{a\in \{0,\ldots,d\}^n\mid a_1+\ldots+a_n\leq d\}.$$ 
The elements of $\M_{\leq d}$ correspond to the monomials in $R_{\leq d}$ via 
$$a=(a_1,\ldots,a_n)\longleftrightarrow x_1^{a_1}\cdots x_n^{a_n}.$$

Build a matrix $M$ whose columns are indexed by the elements of $\M_{\leq d}$ in decreasing order from left to right with respect to $\sigma$. The rows correspond to polynomials of the form $uf$ where $u\in R$ is a monomial, $f\in\mathcal{F}$, and $\deg(uf)\leq d$. Notice that this includes the possibility that $u=1$. In order to associate a row $\ell(g)$ to a polynomial $g$, write $$g=\sum_{a=(a_1,\ldots,a_n)\in \M_{\leq d}} \alpha_a x_1^{a_1}\cdots x_n^{a_n},$$ then $\ell(g)_a=\alpha_a$. The matrix $M$ is called {\bf Macaulay matrix} of $\mathcal{F}$ in degree $d$.

Perform Gaussian elimination on $M$ to obtain a matrix in reduced row echelon form (RREF). 
Any row $\ell=(\ell_{a}\mid a\in \M_{\leq d})$ in the RREF of $M$ corresponds to a polynomial 
$$f_\ell=\sum_{a\in \M_{\leq d}}\ell_a x_1^{a_1}\cdots x_n^{a_n}.$$ If $\deg(f_\ell)<d$, we add new rows to $M$ corresponding to the polynomials $uf_\ell$ where $u$ is a monomial, $\deg(uf_\ell)\leq d$ and $uf_\ell\not\in\rowsp(M)$. Here $\rowsp(M)$ denotes the rowspace of $M$. Repeat the computation of the RREF and the operation of adding new rows, until there are no new rows to add. Denote by $M_d$ the matrix in RREF computed via this algorithm.

It is clear that $\rowsp(M_d)\subseteq (\F)_{\leq d}$. 
For a given $d$, one may have $\rowsp(M_d)\neq (\F)_{\leq d}$. However, it is well-known that, for $d\gg 0$, $\rowsp(M_d)=(\F)_{\leq d}$ and $\rowsp(M_d)$ contains a Gr\"obner basis of $\F$ with respect to $\sigma$. 

\begin{defn}\label{sd}
Let $\F\subseteq R$ be a polynomial system. The {\bf solving degree} of $\F$ with respect to a term order $\sigma$ is the least $d$ such that $\rowsp(M_d)$ contains a Gr\"obner basis of $(\F)$ with respect to $\sigma$. We denote it by $\sd_{\sigma}(\F)$.
\end{defn}

\begin{defn}
Let $\F\subseteq R$ be a polynomial system and let $\sigma$ be a term order. We denote by $\maxGB_\sigma(\F)$ the largest degree of an element of a reduced Gr\"obner basis of $(\F)$ with respect to $\sigma$.
\end{defn}

Since computing the solving degree without computing  a Gr\"obner basis is hard, several authors have introduced other invariants which they use to estimate the solving degree. The main ones are the first and the last fall degree, the degree of regularity, and the Castelnuovo--Mumford regularity.

The {\bf first fall degree} was introduced by Dubois and Gama in \cite{DG10}. Here we use a modified definition from \cite{DS13}. We consider a finite field $\FF_q$ and we work over $B=\FF_q[x_1,\dots,x_n]/(x_1^q,\dots,x_n^q)$. 
Let $f_1,\dots,f_r\in B$ be homogeneous polynomials. If $f_1,\dots,f_r$ are not homogeneous, use $f_1^{\ttop},\dots,f_r^{\ttop}$ instead of $f_1,\dots,f_r$.
We define a $B$-module homomorphism $\varphi$ by letting $\varphi(b_1,\dots,b_r)=\sum_{i=1}^rb_if_i$ for any $(b_1,\dots,b_r)\in B^r$.
We denote by $\Syz(f_1,\dots,f_r)$ the syzygy module of $f_1,\dots,f_r$, that is, the kernel of $\varphi$. 
An element of $\Syz(f_1,\dots,f_r)$ is called a syzygy of $f_1,\dots,f_r$. In other words, a syzygy is a vector $(b_1,\dots,b_r)\in B^r$ such that $\sum_{i=1}^rb_if_i=0$.
The degree of the syzygy $(b_1,\dots,b_r)$ is $$\deg(b_1,\dots,b_r)=\max\{\deg(b_i)+\deg(f_i)\mid 1\leq i\leq r\}.$$
For any $d\in\mathbb{N}$ we define the vector space $\Syz(\mathcal{F})_d$ of homogeneous syzygies of degree $d$, i.e., syzygies $(b_1,\dots,b_r)$ such that $b_1,\ldots,b_r$ are homogeneous and $\deg(b_i)+\deg(f_i)=d$ for all $1\leq i\leq r$. Notice that $\Syz(\mathcal{F})=\oplus_{d\geq 0}\Syz(\mathcal{F})_d$.

An example of syzygy is given by the Koszul syzygies $f_ie_j-f_je_i$, where  $i\neq j$ or by the syzygies coming from the quotient structure of $B$, that is $f_i^{q-1}e_i$. Here $e_i$ denotes the $i$-th element of the canonical basis of $B^r$. These syzygies are called {\bf trivial syzygies}, because they are always present and do not depend on the structure of $f_1,\dots,f_r$, but rather on the ring structure of $B$.
We define the module $\Triv(f_1,\dots,f_r)$ of trivial syzygies of $f_1,\dots,f_r$ as the submodule of $\Syz(f_1,\dots,f_r)$ generated by 
\[\{ f_ie_j-f_je_i: \ 1\leq i<j\leq r \}\cup\{f_i^{q-1}e_i: \ 1\leq i\leq r \}.\]
We define the vector subspace of homogeneous trivial syzygies of degree $d$ as $\Triv(\mathcal{F})_d=\Triv(\mathcal{F})\cap \Syz(\mathcal{F})_d$.

\begin{defn}
Let $\mathcal{F}\subseteq B$ be a polynomial system.
The {\bf first fall degree} of $\mathcal{F}$ is 
\begin{equation*}
d_{\ff}(\mathcal{F})=\min\{d\in\mathbb{N}: \ \Syz(\mathcal{F}^{\mathrm{top}})_{d}/\Triv(\mathcal{F}^{\mathrm{top}})_{d}\neq 0 \}.
\end{equation*}
\end{defn}

The definition of {\bf Castelnuovo--Mumford  regularity} is the most technical and it requires concepts from commutative algebra that we have not defined here. Therefore, we refer the reader to \cite[Section 3.4]{CG21} for the definition of Castelnuovo--Mumford regularity and a discussion on its relation with the solving degree. For a homogeneous system $\F$, we denote by $\reg(\F)$ the Castelnuovo--Mumford  regularity of the ideal generated by $\F$. If the system $\F$ is not homogeneous, we consider the Castelnuovo--Mumford regularity of the associated homogeneous systems $\F^{\ttop}$ and $\F^h$. 

\par The {\bf{degree of regularity}} was introduced by Bardet, Faug\`ere, and Salvy in \cite{Bar04, BFS04}. Let $I$ be a homogeneous ideal of $R=k[x_1,\dots,x_n]$, we recall that for an integer $d\geq0$ we denote by $I_d=I\cap R_d$ the $k$-vector space of homogeneous polynomials of degree $d$ of $I$.
The function $\HF_I(-):\mathbb{N}\rightarrow\mathbb{N}$, $\HF_I(d)=\dim_{k}I_d$ is called the Hilbert function of $I$.  

\begin{defn}\label{def-dregFaugere}
Let $\mathcal{F}\subseteq R$ be a polynomial system.
Assume that $\left(\mathcal{F}^{\mathrm{top}}\right)_d=R_d$ for $d\gg0$. The {\bf degree of regularity} of $\mathcal{F}$ is
\begin{equation*}
d_{\reg}(\mathcal{F})=\min\{d\geq 0\mid \HF_{(\mathcal{F}^{\mathrm{top}})}(d)=\HF_{R}(d)\}=\min\{d\geq 0\mid (\mathcal{F}^{\mathrm{top}})_d=R_d\}.
\end{equation*} 
\end{defn}

The {\bf last fall degree} was introduced by Huang, Kosters, and Yeo in~\cite{HKY15}. In~\cite[Proposition 2.3 and Proposition~2.8]{Huang2018} the authors outline a linear-algebra-based algorithm which computes the solutions of a non-homogeneous polynomial system $\mathcal{F}$ and show that its complexity is controlled by the last fall degree of the system $\mathcal{F}$. 

\begin{defn}\label{defn:VF}\label{defn:lfd} 
Let $\F\subseteq R$ be a polynomial system. Let $i\in\NN$ and let $V_{\F,i}$ be the smallest $k$-vector space such that:
\begin{itemize}
\item $\F\cap R_{\leq i}\subseteq V_{\F,i}$,
\item if $f\in V_{\F,i}$ and $g\in R_{\leq i-\deg(f)}$, then $fg\in V_{\F,i}$.
\end{itemize}
We set also $V_{\F,\infty}=(\F)$ and $V_{\F,-1}=\emptyset$.

The \textbf{last fall degree} of $\F$ is $$d_{\F}=\min\{d\in\NN\cup\{\infty\}\mid f\in V_{\F,\max\{d,\deg(f)\}} \mbox{ for all } f\in (\mathcal{F})\}.$$
\end{defn}

A direct consequence of the existence of a Gr\"obner basis for $(\mathcal{F})$ is that the last fall degree is a natural number, see also~\cite[Proposition~2.6 i]{Huang2018}. 
The concept of last fall degree also allow us to show that the rowspace of a Macaulay matrix after applying the algorithm described at the beginning of this section does not depend on the choice of the term order, provided that it is degree-compatible.
In fact, it follows from Definition~\ref{defn:VF} that $$V_{\F,d}\supseteq\rowsp(M_d)$$ for all $d$.
In the proof of~\cite[Proposition 2.3]{Huang2018} it is stated that, if $\sigma$ is degree-compatible, then 
\begin{equation}\label{eqn:rowspMM}
V_{\F,d}=\rowsp(M_d)
\end{equation} for all $d\in\ZZ_{\geq 0}$. A proof of \eqref{eqn:rowspMM} is given in~\cite[Theorem~1]{Gorla2018}.

\begin{cor}\label{cor:to_indep}
$\rowsp(M_d)$ does not depend on the choice of $\sigma$, if $\sigma$ is degree-compatible.
\end{cor}

The next example shows that, in general, equality \eqref{eqn:rowspMM} and Corollary~\ref{cor:to_indep} do not hold for a term order which is not degree-compatible.

\begin{ex}
Let $d\geq 3$ be an integer. Let $\mathcal{F}=\{x-y^{d-1},x-y^d\}$ and let $\sigma$ be the lexicographic order on $k[x,y]$ with $x>y$. Let $$f=xy-x=y(x-y^{d-1})-(x-y^d)\in V_{\mathcal{F},d}.$$ By the definition of last fall degree, one also has that $xf\in V_{\mathcal{F},d}$. 
The rows of the matrix $M_d$ correspond to the polynomials $x-y^{d-1}, x(x-y^{d-1}), y(x-y^{d-1})$, and $x-y^d$. One can check that all the rows of the RREF of $M_d$ correspond to polynomials of degree $d$, therefore $$\rowsp(M_d)=\langle x-y^{d-1}, x(x-y^{d-1}), y(x-y^{d-1}), y^{d-1}-y^d\rangle.$$ In particular, $xf\not\in\rowsp(M_d)$ since it contains the monomial $x^2y$ in its support.
\end{ex}

\section{Equivalent definitions for the last fall degree}\label{sect:lfd}

In this section, we explore some properties of the last fall degree and discuss different definitions for this invariant.
We start by observing that, by Definition~\ref{defn:lfd}, one has 
$$V_{\F,d}=(\mathcal{F})_{\leq d}$$
for any polynomial system $\mathcal{F}\subseteq R$ and any $d\geq d_{\F}$.

\begin{prop}\label{prop:eq_ideal}
One has $$V_{\F,d_{\F}-1}\neq(\mathcal{F})_{\leq d_{\F}-1}.$$
In particular $$d_{\F}=\min\{d\mid V_{\F,e}=(\mathcal{F})_{\leq e}\mbox{ for all } e\geq d\}.$$

\end{prop}

\begin{proof}

By Definition~\ref{defn:lfd} there exists $f\in (\mathcal{F})$ such that $f\not\in V_{\F,\max\{d_{\F}-1,\deg(f)\}}$. 
If $\deg(f)\geq d_{\F}$, then $f\in V_{\F,\deg(f)}$ by Definition~\ref{defn:lfd}. This contradicts the assumption that $f\not\in V_{\F,\max\{d_{\F}-1,\deg(f)\}}=V_{\F,\deg(f)}$.
Therefore $\deg(f)\leq d_{\F}-1$. Then $f\in (\mathcal{F})_{\leq d_{\F}-1}$ and $f\not\in V_{\F,d_{\F}-1}$, in particular $V_{\F,d_{\F}-1}\neq (\mathcal{F})_{\leq d_{\F}-1}.$ 
\end{proof}

\begin{rmk}
The same polynomial $f$ as in the proof of Proposition~\ref{prop:eq_ideal} satisfies $f\not\in V_{\F,d}$ for any $d\leq d_{\F}-1$, since $V_{\F,d}\subseteq V_{\F,d_{\F}-1}$. 
Therefore, for any $\deg(f)\leq d\leq d_{\F}-1$ one has $$V_{\F,d}\neq (\mathcal{F})_{\leq d}.$$
\end{rmk}

Notice that the definition of last fall degree from Proposition~\ref{prop:eq_ideal} is not computationally more efficient than Definition~\ref{defn:lfd}. In fact, the standard way of computing $(\F)_{\leq d}$ is by computing a (part of a) Gr\"obner basis of $\F$. 

We now discuss another equivalent definition of last fall degree, coming from~\cite{Huang2018}. In~\cite[Proposition 2.6]{Huang2018} it is shown that $$d_\F=\min\{d\in\NN\mid V_{\F,e}=V_{\F,e+1}\cap R_{\leq e}\mbox{ for all } e\geq d\}.$$ 
We now study the set of $d$'s such that $V_{\F,d}=V_{\F,d+1}\cap R_{\leq d}$. 

\begin{prop}
Let $d\in\NN$. If $V_{\F,d}=(\mathcal{F})_{\leq d}$, then $V_{\F,d}=V_{\F,d+1}\cap R_{\leq d}.$ In other words, 
\begin{equation}\label{eq:inclusion}
\{d\in\NN\mid V_{\F,d}=(\mathcal{F})_{\leq d}\}\subseteq\{d\in\NN\mid V_{\F,d}=V_{\F,d+1}\cap R_{\leq d}\}.
\end{equation} 
\end{prop}

\begin{proof}
If $V_{\F,d}=(\mathcal{F})_{\leq d}$, then $$V_{\F,d}=(\mathcal{F})\cap R_{\leq d}=((\mathcal{F})\cap R_{\leq d+1})\cap R_{\leq d}\supseteq V_{\F,d+1}\cap R_{\leq d}\supseteq V_{\F,d}$$ where the last inclusion holds for all $d$ by Definition~\ref{defn:VF}.
\end{proof}

Since the equality $V_{\F,d}=V_{\F,d+1}\cap R_{\leq d}$ can be checked for any given $d$ by checking whether $\rowsp(M_{d+1})\cap R_{\leq d}=\rowsp(M_{d})$, one may think that the definition from~\cite[Proposition 2.6]{Huang2018} is more computationally friendly than that from Proposition~\ref{prop:eq_ideal}. Unfortunately, one cannot hope to compute $d_{\F}$ by progressively increasing the value of $d$ until one finds a $d$ for which $V_{\F,d}=V_{\F,d+1}\cap R_{\leq d}$.
In fact, in Example~\ref{ex:intervals} we show that there are values of $d$ for which $V_{\F,d-1}=V_{\F,d}\cap R_{\leq d-1}$ but $V_{\F,d}\neq V_{\F,d+1}\cap R_{\leq d}$. 

Notice moreover that, while the definition from Proposition~\ref{prop:eq_ideal} and that from~\cite[Proposition 2.6]{Huang2018} are equivalent, the inclusion in \eqref{eq:inclusion} may be strict. In fact, even if $[d_\F,+\infty)$ is the largest right-unbounded interval contained in either of the sets from \eqref{eq:inclusion}, each of the two sets can contain smaller values of $d\in\NN$. The next example shows that this can indeed occur.

\begin{ex}\label{ex:intervals}
Let $\F=\{x^2-x, xy-1, w^6-w, w^5z^5-1\}\subseteq R=k[w,x,y,z]$ with $\chara k=0$ or $\chara k\gg 0$. It is easy to check that $(\F)=(x-1,y-1,w^5-1,z^5-1)$.
Since $\F$ contains no elements of degree 0 or 1, then $$V_{\F,0}=V_{\F,1}=0.$$
Moreover $$V_{\F,2}=\langle x^2-x, xy-1 \rangle,$$ since the Macaulay matrix associated to $\F$ in degree 2 with respect to a degree-compatible term order is in RREF and contains only rows of degree 2.
By repeatedly performing Gaussian elimination and adding rows whenever the matrix contains new rows of degree smaller than 3, one can check that $x-1, y-1\in V_{\F,3}$, hence
$$V_{\F,3}\supseteq (x-1)R_{\leq 2}+(y-1)R_{\leq 2}=(\F)_{\leq 3}.$$
Therefore $$V_{\F,d}=V_{\F,d+1}\cap R_{\leq d}$$ for $d=0,1$, but $$x-1,y-1\in (V_{\F,3}\cap R_{\leq 2})\setminus V_{\F,2}.$$
Since $(\F)_{\leq 4}=(x-1,y-1)_{\leq 4}$, then $$V_{\F,3}=V_{\F,4}\cap R_{\leq 3}.$$
One can also check that $V_{\F,5}=V_{\F,4}\cdot R_{\leq 1}$, hence $$V_{\F,4}=V_{\F,5}\cap R_{\leq 4}.$$
Notice moreover that, for $0\leq d\leq 5$ one has $$V_{\F,d}=(\F)_{\leq d} \mbox{ if and only if } d\neq 1,2,5.$$

Now $V_{\F,6}=V_{\F,5}\cdot R_{\leq 1}+\langle w^6-w\rangle$ and $V_{\F,d}=V_{\F,6}\cdot R_{\leq d-6}$ for $6\leq d\leq 9$. Finally, $V_{\F,10}=V_{\F,9}\cdot R_{\leq 1}+\langle w^5z^5-1\rangle$. Therefore 
$$V_{\F,d}=V_{\F,d+1}\cap R_{\leq d}$$ for $5\leq d\leq 9$.
Moreover, one can check that $w^5-1\not\in V_{\F,d}$ for $d\leq 10$, hence $$V_{\F,d}\neq (\F)_{\leq d} \mbox{ for } 6\leq d\leq 10.$$ 
One can verify by direct computation that $w^5-1,z^5-1\in V_{\F_{11}}$. Therefore $$V_{\F,10}\neq V_{\F,11}\cap R_{\leq 10}.$$ Moreover, $$V_{\F,d}=(\F)_{\leq d} \mbox{ for } d\geq 11,$$ therefore also
$$V_{\F,d}=V_{\F,d+1}\cap R_{\leq d}\mbox{ for } d\geq 11.$$ 
In fact, this shows that $d_{\F}=11$.
In the next table, we summarize for which values of $d$ each equality holds.

\bgroup
\def\arraystretch{1.5}
\begin{tabular}{l|c|c|c|c|c|c|c|c|c|c|c|c|c|c}
$d$ & 0 & 1 & 2 & 3 & 4 & 5 & 6 & 7 & 8 & 9 & 10 & 11 & 12 & \ldots  \\
\hline
$V_{\F,d}=V_{\F,d+1}\cap R_{\leq d}$ & \cmark & \cmark & \xmark & \cmark & \cmark & \cmark & \cmark & \cmark & \cmark & \cmark & \xmark & \cmark & \cmark & \ldots \\
\hline
$V_{\F,d}=(\F)_{\leq d}$ & \cmark & \xmark & \xmark & \cmark & \cmark & \xmark & \xmark & \xmark & \xmark & \xmark & \xmark & \cmark & \cmark & \ldots 
\end{tabular}
\egroup \\ \\ \noindent
\end{ex}

\subsection{Degree falls and last fall degree}

In this subsection, we give a mathematical formulation for the notion of degree falls. This allows us to prove that the last fall degree is the largest degree in which a degree fall occurs. This is also the largest degree for which a degree fall occurs for an element of a Gr\"obner basis. We start by defining degree falls.

\begin{defn}
Let $\F\subseteq R$ be a polynomial system. For any $f\in (\mathcal{F})$ let
\[d_f=\min\{d\in\mathbb{N}: f\in V_{\mathcal{F},d}\}.\]
\end{defn}

\noindent We always have that $d_f\geq \deg(f)$. 

\begin{defn}\label{defn:degfalls}
If $d_f>\deg(f)$, then we say that $f$ has a {\bf degree fall} in degree $d_f$. 
We say that $\mathcal{F}$ has a {\bf degree fall} if there is an $f\in(\mathcal{F})$ such that $f$ has a degree fall. Else we say that  $\mathcal{F}$ has {\bf no degree falls}.
\end{defn}

For example, a homogeneous system has no degree falls. Moreover, its last fall degree is zero. More in general, one can prove the following.

\begin{lemma}\label{lemma:lfd=0}
Let $\mathcal{F}\subseteq R$ be a polynomial system. Then $\mathcal{F}$ has no degree falls if and only if $d_{\mathcal{F}}=0$.
\end{lemma}

\begin{proof}
If $d_{\F}=0$, then for all $f\in(\F)$ we have $f\in V_{\F,\deg(f)}$. Therefore $\F$ has no degree falls. Conversely, assume that $\mathcal{F}$ has no degree falls and let $f\in(\F)$.  
Since $d_f=\deg(f)$, then $f\in V_{\mathcal{F},\deg(f)}$. This proves that $d_{\mathcal{F}}=0$.
\end{proof}

We can now prove that the last fall degree is the largest degree in which a degree fall occurs.

\begin{thm}\label{d_f}
Let $\mathcal{F}\subseteq R$ be a polynomial system. If $\mathcal{F}$ has a degree fall, then 
\[
d_{\mathcal{F}}=\max\{d_f: \ f\in (\mathcal{F}), d_f>\deg(f)\}.
\]
\end{thm}

\begin{proof}
Let $f\in\mathcal{F}$. If $d_f>\deg(f)$, then $f\not\in V_{\mathcal{F},d}$ for any $d<d_f$. In particular, $d_\mathcal{F}>d_f-1$ by the definition of last fall degree. This implies that $$d_{\mathcal{F}}\geq\sup\{d_f: \ f\in (\mathcal{F}), d_f>\deg(f)\}.$$
Notice that, by the finiteness of the last fall degree, the supremum of the set $\{d_f: \ f\in (\mathcal{F}), d_f>\deg(f)\}$ is in fact a maximum. 

We now prove the reverse inequality.
Let $\mu=\max\{d_f: \ f\in (\mathcal{F}), d_f>\deg(f)\}$ and fix $f\in(\mathcal{F})$. If $d_f>\deg(f)$ we have $f\in V_{\mathcal{F},d_f}\subseteq V_{\mathcal{F},\mu}$ and $\mu=\max\{\mu,\deg(f)\}$ since $\mu\geq d_f>\deg(f)$. If $d_f=\deg(f)$, then $f\in V_{\mathcal{F},\deg(f)}\subseteq V_{\mathcal{F},\max\{\mu,\deg(f)\}}$. Therefore, for all $f\in(\mathcal{F})$, $f\in V_{\mathcal{F},\max\{\mu,\deg(f)\}}$, so $\mu\in\{d\in\NN\cup\{\infty\}\mid f\in V_{\F,\max\{d,\deg(f)\}} \mbox{ for all } f\in (\mathcal{F})\}.$ It follows that $d_\mathcal{F}\leq\mu$.
\end{proof}

It is worth noticing that the last fall degree is the largest degree of a degree fall occuring for an element of a Gr\"obner basis of $\mathcal{F}$ with respect to a degree-compatible term order. 

\begin{thm}\label{gbasis}
Let $\mathcal{F}\subseteq R$ be a polynomial system, and let $\{g_1,\dots,g_s\}$ be a Gr\"obner basis of $\mathcal{F}$ with respect to a degree compatible term order. If $\mathcal{F}$ has a degree fall, then $$d_{\mathcal{F}}=\max\{d_{g_i}: \ \deg(g_i)<d_{g_i}\}.$$
In particular, if $d_{g_i}=\deg(g_i)$ for $i=1,\ldots,s$, then $d_{\mathcal{F}}=0$.
\end{thm}

\begin{proof}
Up to reordering, we may assume without loss of generality that $\{g_1,\dots,g_t\}$ are a minimal Gr\"obner basis of $\mathcal{F}$ for some $t\leq s$. Let $\delta_{\mathcal{F}}=\max\{d_{g_i} : d_{g_i}>\deg(g_i), 1\leq i\leq t\}$.
Then 
\begin{equation*}
\begin{split}
d_\mathcal{F}&=\max\{d_f: \ f\in (\mathcal{F}), d_f>\deg(f)\}\\&\geq
\max\{d_{g_i} : d_{g_i}>\deg(g_i), 1\leq i\leq s\}\geq \delta_{\mathcal{F}},
\end{split}
\end{equation*}
where the first equality follows from Theorem~\ref{d_f}. Therefore, it suffices to show that the statement holds for $\{g_1,\dots,g_t\}$ a minimal Gr\"obner basis.

To prove the reverse inequality, let $f\in(\mathcal{F})$ and write $$f=\sum_{i\in J} h_ig_i$$ with $J\subseteq\{1,\ldots,t\}$, $\deg(h_i)\leq\deg(f)-\deg(g_i)$ and $h_i\neq 0$ for all $i\in J$. Fix $i\in J$ and let $u_i=\max\{\deg(f),d_{g_i}\}$. Then $g_i\in V_{\mathcal{F},d_{g_i}}\subseteq V_{\mathcal{F},u_i}$ and $h_i\in R_{\leq \deg(f)-\deg(g_i)}\subseteq R_{\leq u_i-\deg(g_i)}$. It follows that $h_ig_i\in V_{\mathcal{F},u_i}$. Therefore $f\in V_{\mathcal{F},\max\{\deg(f),d_{g_i} :\, i\in J\}}$. Notice that if $i\in J$ and $\deg(g_i)=d_{g_i}$, then $d_{g_i}\leq\deg(f)$. Therefore $$\max\{\deg(f),d_{g_i} :\, i\in J\}\leq \max\{\deg(f),\delta_{\mathcal{F}}\} \;\mbox{ and }\; f\in V_{\mathcal{F},\max\{\deg(f),\delta_{\mathcal{F}}\}}.$$
This shows that $d_\mathcal{F}\leq \delta_{\mathcal{F}}$.
\end{proof}

\begin{rmk}
Theorem~\ref{gbasis} implies that the largest degree of a degree fall occuring for an element of a Gr\"obner basis of $\mathcal{F}$ with respect to a degree-compatible term order $\sigma$ is independent of the choice of $\sigma$ and of the Gr\"obner basis.
\end{rmk}

\section{Solving degree and last fall degree}\label{section:sdeLFD}

It is clear by definition that the solving degree of a system with respect to a fixed term order is bounded from below by the largest degree of an element of a Gr\"obner basis with respect to that term order. In addition, in~\cite[Theorem~1]{Gorla2018} it was proved that the solving degree with respect to a degree-compatible term order is bounded from below by the last fall degree. 
The next theorem clarifies the relation between the last fall degree and the solving degree, by providing a direct comparison. 
From a computational point of view, the theorem allows us to split estimates for the solving degree in two parts: an estimate on the degrees of the elements in a reduced Gr\"obner basis and an estimate on the last fall degree. Notice that, while the degrees of the elements in a Gr\"obner basis depend on the choice of the order and on the system of coordinates, the last fall degree does not.

\begin{thm}\label{sd&lfd}
Let $\mathcal{F}\subseteq R$ be a polynomial system. Let $\sigma$ be a degree-compatible term order. Then
$$\sd_{\sigma}(\mathcal{F})=\max\{d_{\mathcal{F}},\maxGB_{\sigma}(\mathcal{F})\}.$$
\end{thm}

\begin{proof}
If $\mathcal{F}$ is homogeneous, then $d_{\mathcal{F}}=0$. The thesis follows from observing that $\sd_{\sigma}(\mathcal{F})=\maxGB_{\sigma}(\mathcal{F})$, whenever $\mathcal{F}$ is homogeneous. 

Suppose now that $\mathcal{F}$ is not homogeneous.
The inequality $\sd_{\sigma}(\mathcal{F})\geq d_{\mathcal{F}}$ follows from~\cite[Theorem~1]{Gorla2018}. From the definition of solving degree, $\sd_{\sigma}(\mathcal{F})\geq\maxGB_{\sigma}(\mathcal{F})$. Therefore, it suffices to show that $\sd_{\sigma}(\mathcal{F})\leq\max\{d_{\mathcal{F}},\maxGB_{\sigma}(\mathcal{F})\}.$ 
Let $\rowsp(M_d)$ be the vector space generated by the rows of the matrix obtained from the Macaulay matrix $M$ in degree $d$ after applying the algorithm described in Section~\ref{sect:prelim}.
By~\cite[Theorem~1]{Gorla2018}, $\rowsp(M_d)=V_{\mathcal{F},d}$ for every $d\in\mathbb{N}$.

Let $D=\max\{d_{\mathcal{F}},\maxGB_{\sigma}(\mathcal{F})\}$. Then 
$$\rowsp(M_D)=V_{\mathcal{F},D}=(\mathcal{F})\cap R_{\leq D},$$
where the second equality follows from the definition of last fall degree, since $D\geq d_{\mathcal{F}}$. Let $g$ be an element of the reduced Gr\"obner basis of $\mathcal{F}$ with respect to $\sigma$. Then $g\in \rowsp(M_D)$, therefore $g$ is a linear combination of the rows of the matrix $M_D$ obtained from the Macaulay matrix in degree $D$ after applying the algorithm from Section~\ref{sect:prelim}. Since $M_D$ is in reduced row-echelon form, the leading terms of the rows of $M_D$ are pairwise distinct.
 
We claim that, up to scalar multiples, $g$ is a row of $M_D$. In fact, if $g$ is a combination of more than one row, since $M_D$ is in reduced row-echelon form, then $g$ contains in its support one or more monomials that appear as leading term of a row of $M_D$, beyond its leading term. However this is impossible, since $g$ is an element of a reduced Gr\"obner basis and each row of $M_D$ corresponds to an element of $(\mathcal{F})$, hence its leading term is divisible by a leading term of an elements of the reduced Gr\"obner basis. Hence a scalar multiple of $g$ appears among the rows of $M_D$, so $\sd_{\sigma}(\mathcal{F})\leq D$.
\end{proof}

\begin{rmk}
Notice that while $\sd_{\sigma}(\mathcal{F})$ and $\maxGB_{\sigma}(\mathcal{F})$ depend in general on the term order $\sigma$ and on the system of coordinates, the last fall degree of $\mathcal{F}$ does not depend on either of them.
\end{rmk}

The next corollary follows from combining Theorem \ref{gbasis} and Theorem \ref{sd&lfd}.

\begin{cor}
Let $\mathcal{F}\subseteq R$ be a polynomial system. Let $\sigma$ be a degree-compatible term order and let $\{g_1,\dots,g_s\}$ be a reduced Gr\"obner basis of $\mathcal{F}$ with respect to $\sigma$. Then
$$\sd_{\sigma}(\mathcal{F})=\max\{d_{g_1},\ldots,d_{g_s}\}.$$
\end{cor}

The next example shows that the conclusion of Theorem~\ref{sd&lfd} is in general false for term orders which are not degree-compatible.

\begin{ex}[\cite{Gorla2018}, Example~2]
Let $d\geq 2$ be an integer. Let $\F=\{x_1-x_1x_3^{d-1},x_2-x_3^d\}$ and let $\sigma$ be the lexicographic order on $k[x_1,x_2,x_3]$ with $x_1>x_2>x_3$. The elements of $\F$ are a Gr\"obner basis of the ideal that they generate and $\sd_{\sigma}(\F)=d$. Let $$f=x_1x_3-x_1x_2=x_3(x_1-x_1x_3^{d-1})-x_1(x_2-x_3^d)\in V_{\F,d+1}.$$ Since $f\not\in V_{\F,d}=\langle x_1-x_1x_3^{d-1},x_2-x_3^d \rangle$, then $d_{\F}\geq d+1.$ Therefore 
$$\sd_{\sigma}(\F)=d<\max\{d_{\F},\maxGB_{\sigma}(\F)\}.$$
\end{ex}

\section{How does the first fall degree relate to the other invariants?}\label{sect:ffd}

The first fall degree is used by many authors as an estimate for the solving degree. While this is justified by a heuristic argument, this does not always provide a reliable estimate. In fact, it is easy to produce examples where the first fall degree and the solving degree are far apart, see \cite{DS13}.
In this section, we explore further the relation of the first fall degree with the other algebraic invariants related to the solving degree. We provide examples that show that the following are possible:
\begin{itemize}
\item $d_{\ff}(\mathcal{F})>\sd(\mathcal{F})$ or $d_{\ff}(\mathcal{F})<\sd(\mathcal{F})$,
\item $d_{\ff}(\mathcal{F})>d_{\F}$ or $d_{\ff}(\mathcal{F})<d_{\F}$, 
\item $d_{\ff}(\mathcal{F})>d_{\reg}(\F)$ or $d_{\ff}(\mathcal{F})<d_{\reg}(\F)$,
\item $d_{\ff}(\mathcal{F})>\reg(\F^h)$ or $d_{\ff}(\mathcal{F})<\reg(\F^h)$,
\end{itemize}
In addition, our examples show that the difference between the first fall degree and each of the other invariants can be arbitrarily small or arbitrarily large.

\begin{ex}
Let $q>3$ be odd. Fix any degree-compatible term order. Then $\mathcal{F}=\{x_1x_2+x_2, x_2^2-1, x_1^{q-1}-1\}\subseteq\FF_q[x_1,x_2]$ has Gr\"obner basis $\{x_1+1,x_2^2-1\}$. One can compute $$\mathcal{F}^{\mathrm{top}}=\{x_1x_2, x_2^2, x_1^{q-1}\},$$ with non-trivial syzygies
$$(x_2,-x_1,0), (x_1^{q-1},0,0), (x_2^{q-1},0,0), (x_1^{q-2},0,-x_2), (0,x_2^{q-2},0), (0,0,x_1)$$ in $B$. 
Since $\mathcal{F}^{\mathrm{top}}$ has no syzygies in degree smaller than $3$, then $d_{\ff}(\mathcal{F})=3$.
It follows from the definition that $d_{\reg}(\F)=q-1$.
One can show by direct computation that $\sd(\F)=3$.
Then $d_{\F}=3$ by Theorem~\ref{sd&lfd}.

Let $x_0$ be the homogenizing variable and consider $S=\FF_q[x_0,x_1,x_2]$. The Castelnuovo--Mumford regularity of $\F^h$ can be obtained from its minimal free resolution $$0\rightarrow S(-q-2)\rightarrow\begin{array}{c} S(-q-1) \\ \oplus \\ S(-q) \\ \oplus \\ S(-4)\end{array} \rightarrow \begin{array}{c} S(-q+1) \\ \oplus \\ S(-2)^2\end{array}\rightarrow (\F^h)\rightarrow 0$$
where $\{x_1x_2+x_2x_0, x_2^2-x_0^2, x_1^{q-1}-x_0^{q-1},x_1x_0^2+x_0^3\}$ is a Gr\"obner basis of $\F^h$ with respect to the degree reverse lexicographic term order with $x_1>x_2>x_0$. A minimal system of generators of the first syzygy module of $x_1^{q-1}-x_0^{q-1},x_1x_2+x_2x_0,x_2^2-x_0^2$ is $$\left\{\left(x_2,\sum_{j=0}^{q-2}(-1)^{j+1} x_0^jx_1^{q-2-j},0\right),(-x_2^2+x_0^2,0,x_1^{q-1}-x_0^{q-1}),(0,-x_2^2+x_0^2,x_1x_2+x_2x_0)\right\}$$ and the second syzygy module is $\langle(-x_2^2+x_0^2,-x_2,\sum_{j=0}^{q-2}(-1)^{j+1} x_0^jx_1^{q-2-j})\rangle$.
Therefore we have shown that
$$d_{\ff}(\mathcal{F})=\sd(\mathcal{F})=d_{\F}=3< q-1=d_{\reg}(\F)<q=\reg(\F^h).$$
\end{ex}

\begin{ex}\label{ex:2}
Let $q>3$ be odd. Fix any degree-compatible term order and let $\mathcal{F}=\{x_1x_2+x_2, x_2^2-1, x_3^{q-1}-1,x_1^q-x_1\}\subseteq\FF_q[x_1,x_2,x_3]$. Notice that, since $x_2^2-1\mid x_2^q-x_2$ and $x_3^{q-1}-1\mid x_3^q-x_3$, the invariants of $\F$ and those of $\F\cup\{x_2^q-x_2, x_3^q-x_3\}$ are the same. One has $$\mathcal{F}^{\mathrm{top}}=\{x_1^q,x_1x_2, x_2^2, x_3^{q-1}\}$$
and $\{x_1x_2, x_2^2, x_3^{q-1}\}\subseteq B$ has non-trivial syzygies
$$(x_2,-x_1,0), (x_1^{q-1},0,0), (x_2^{q-1},0,0), (0,x_2^{q-2},0), (0,0,x_3)\in B^3.$$ Since $\mathcal{F}^{\mathrm{top}}$ has no syzygies in degree smaller than $3$, then $d_{\ff}(\mathcal{F})=3$. Moreover, one can check by direct computation that $d_{\reg}(\F)=2q-2$.

The Macaulay matrix of $\F$ in degree $2$ is already in RREF, therefore $x_1+1\not\in V_{\F,2}$. Moreover, $x_1+1=x_2(x_1x_2+x_2)-(1+x_1)(x_2^2-1)\in V_{\F,3}$, hence $x_1+1$ has a degree fall in degree $3$. Since $\{x_1+1,x_2^2-1,x_3^{q-1}-1\}$ is a Gr\"obner basis of $(\F)$, then $d_{\F}=3$ by Theorem~\ref{gbasis}. Moreover, the solving degree is $q-1$ by Theorem~\ref{sd&lfd}.
Therefore we have
$$d_{\ff}(\mathcal{F})=d_{\F}=3<q-1=\sd(\mathcal{F})<2q-2=d_{\reg}(\mathcal{F}).$$ 
\end{ex}

\begin{ex}\label{ex:dreg<reg}
Let $q$ be an odd prime power and let $\F=\{x_1^2-1,x_2^2-1\}\subseteq \FF_q[x_1,x_2]$. Notice that $x_i^2-1\mid x_i^q-x_i$ for $i=1,2$, hence $\F$ and $\F\cup\{x_1^q-x_1,x_2^q-x_2\}$ have the same invariants.
The system $\F^h=\{x_1^2-x_0^2,x_2^2-x_0^2\}$ is a regular sequence, in particular $\reg(\F^h)=3$. Moreover, the non-trivial syzygies of $\F^{\ttop}$ in $B=R/(x_1^q,x_2^q)$ are generated by $(x_1^{q-2},0),(0,x_2^{q-2})$. This shows that $d_{\ff}(\mathcal{F})=q$. Since $(x_1^2,x_2^2)_2\neq \FF_q[x_1,x_2]_2$ and $(x_1^2,x_2^2)_3= \FF_q[x_1,x_2]_3$, then $d_{\reg}(\F)=3$. Since the elements of $\F$ are a Gr\"obner basis with respect to any term order, then $d_{\mathcal{F}}=0$ by Theorem \ref{gbasis} and $\sd(\F)=2$ by Theorem \ref{sd&lfd}. Therefore 
$$d_{\mathcal{F}}=0<2=\sd(\F)<\reg(\mathcal{F}^h)=d_{\reg}(\F)=3<q=d_{\ff}(\F).$$
\end{ex}

\begin{ex}\label{ex:4}
Let $\F=\{x^2-x, xy-1, w^q-w, w^{q-1}z^{q-1}-1\}\subseteq R=\FF_q[w,x,y,z]$, $q\neq 2$. It is easy to check that $\{x-1,y-1,w^{q-1}-1,z^{q-1}-1\}$ is a Gr\"obner basis of $(\F)$ with respect to any term order. The element $(y,-x,0,0)$ is a non-trivial syzygy of $\F^{\ttop}=\{x^2,xy,w^{q},w^{q-1}z^{q-1}\}$ of degree $3$. Since $\F^{\ttop}$ has no syzygy in degree smaller than $3$, then $d_{\ff}(\F)=3.$

One can verify by direct computation that $x-1,y-1\in V_{\F,3}\setminus V_{\F,2}$ and $w^{q-1}-1,z^{q-1}-1\in V_{\F,2q-1}\setminus V_{\F,2q-2}$, see also Example~\ref{ex:intervals}. By Theorem \ref{gbasis} we have that $\sd(\F)=2q-1$, hence also $d_{\F}=2q-1$ by Theorem~\ref{sd&lfd}. Hence 
$$d_{\ff}(\F)=3<2q-1=\sd(\F)=d_{\F}.$$
\end{ex}

The examples show that some caution is needed when using the first fall degree as a heuristic estimate for the solving degree. Since the first fall degree is defined in $B$ while the other invariants are defined in $R$, in order to have a meaningful comparison it seems appropriate to add the field equations to the system when computing the invariants over $R$. But even in that case, Examples \ref{ex:2} and \ref{ex:dreg<reg} show that the first fall degree may be far from the other invariants of the system, including its solving degree.

Our examples suggest two possible causes for this behavior: When there are nontrivial syzygies with a large difference in their degrees, the first fall degree may be much smaller than the other invariants, including the solving degree. This is the case, e.g., in Examples \ref{ex:2} and \ref{ex:4}. Moreover, the fact that the trivial syzygies do not influence the first fall degree may make the first fall degree artificially large, creating a gap between the first fall degree and the other invariants. This is the case, e.g., in Example \ref{ex:dreg<reg}.

\section{Degree of regularity and Castelnuovo--Mumford regularity}\label{sect:reg}

Both the degree of regularity and the Castelnuovo--Mumford regularity provide upper bounds on the solving degree of a system of polynomial equations. The main result of this section is an inequality between the degree of regularity and the Castelnuovo--Mumford regularity. Throughout the section we assume that $(\F^{\ttop})_d=R_d$ for $d\gg 0$, as this is necessary for defining the degree of regularity.

The following relation between the solving degree of a system $\mathcal{F}$ of polynomial equations and the Castelnuovo--Mumford regularity of the system $\mathcal{F}^h$ obtained from $\mathcal{F}$ by homogenizing its equations with respect to a new variable was proved in \cite{CG21}. In case $\mathcal{F}$ is homogeneous, then $\mathcal{F}^h=\mathcal{F}$. 

\begin{thm}[{\cite[Theorem 9, Theorem 10, and Theorem 11]{CG21}}]\label{CMreg}
Let $\mathcal{F}\subseteq \FF_q[x_1,\dots,x_n]$ be a polynomial system which contains the field equations.  
Then
$$\sd(\mathcal{F})\leq\reg(\mathcal{F}^h).$$
\end{thm}

The next result was proved by Semaev and Tenti. To the extent of our knowledge, this is the first proven upper bound for the solving degree of a system in terms of its degree of regularity. Notice however that the solving degree considered by Semaev and Tenti may be different from the one considered in this paper, as Semaev and Tenti work in the quotient ring $\FF_q[x_1,\dots,x_n]/(x_1^q-x_1,\ldots,x_n^q-x_n)$. For more details see the doctoral thesis of Tenti \cite{T}, in particular the discussion above \cite[Theorem 2.1]{ST} and \cite[Corollary 3.67]{T}.

\begin{thm}[{\cite[ Theorem 2.1]{ST}}]\label{semaevtenti}
Let $\mathcal{F}=\{f_1,\ldots,f_r,x_1^q-x_1,\ldots,x_n^q-x_n\}\subseteq \FF_q[x_1,\dots,x_n]$ be a polynomial system.
If $d_{\reg}(\F)\geq\max\{q,\deg(f_1),\ldots,\deg(f_r)\}$, then $$\sd(\mathcal{F})\leq 2d_{\reg}(\F)-2.$$
\end{thm}

Because of Theorem~\ref{CMreg} and Theorem~\ref{semaevtenti}, both the degree of regularity of $\mathcal{F}$ and the Castelnuovo--Mumford regularity of $\mathcal{F}^h$ produce estimates for the solving degree of $\mathcal{F}$. Hence it is interesting to relate these invariants to each other. 

\begin{thm}\label{thm:reg}
Let $\mathcal{F}\subseteq R=k[x_1,\dots,x_n]$ be a polynomial system. Assume that $(\F^{\ttop})_d=R_d$ for $d\gg 0$.
Then $$d_{\reg}(\F)\leq \reg(\F^h).$$
\end{thm}

\begin{proof}
Since $d_{\reg}(\F)=\reg_R(\F^{\ttop})$ by \cite[Proposition 2]{CG21}, it suffices to show that $\reg_R(\F^{\ttop})\leq\reg_S(\F^h)$, where $S=R[t]$. Write $(\F^h)=J+tH$, where $t\nmid 0$ modulo $J$. Then $(\F^h):t=J+H$ and $S/(\F^h)+(t)=S/(\F^{\ttop})+(t)\cong R/(\F^{\ttop})$. Therefore $$\reg_R(\F^{\ttop})=\reg_S((\F^{\ttop})+(t)).$$
Moreover, one has the short exact sequence $$0\rightarrow S/J+H(-1)\rightarrow S/(\F^h)\rightarrow S/(\F^{\ttop})+(t)\rightarrow 0.$$
By the Mapping Cone Construction one has that \begin{equation}\label{eqn:mappingcone}
\reg_S((\F^{\ttop})+(t))\leq\max\{\reg_S(\F^h),\reg_S(J+H)\}.
\end{equation}
Moreover, by \cite[Theorem 3.1]{CHHKTT}
\begin{equation}\label{eqn:CHHKTT}
\reg_S(\F^h)\in\{\reg_S((\F^{\ttop})+(t)),\reg_S(H+J)+1\}.
\end{equation}
If $\reg_S(\F^h)=\reg_S(H+J)+1$, then by \eqref{eqn:mappingcone} $\reg_S((\F^{\ttop})+(t))\leq \reg_S(H+J)+1=\reg_S(\F^h)$. If $\reg_S(\F^h)\neq\reg_S(H+J)+1$, then by \eqref{eqn:CHHKTT} $\reg_S((\F^{\ttop})+(t))=\reg_S(\F^h)$.
\end{proof}

\bibliographystyle{amsplain}

\end{document}